\documentclass[a4paper, 10 pt, conference]{ieeeconf}
\pdfoutput=1
\usepackage{etex}
\usepackage{pgfplots}
\pgfplotsset{compat=newest}
\usepackage{tikz}
\usetikzlibrary{arrows,matrix,positioning}
\usepackage[utf8]{inputenc}
\usepackage[innermargin=0.545in,outermargin=0.65in,top=0.545in,bottom=.7in]{geometry}
\usepackage[english]{babel}
\usepackage[T1]{fontenc}
\usepackage{epsfig}
\usepackage{amsmath, amssymb, amsbsy}
\usepackage{mathdots}
\usepackage{xspace}
\usepackage[noend]{algpseudocode}
\usepackage{algorithm}
\usepackage{algorithmicx}
\usepackage{color}
\usepackage{cite}
\usepackage{booktabs}
\usepackage{verbatim}
\usepackage[colorlinks=true,citecolor=blue,linkcolor=blue,urlcolor=blue]{hyperref}
\usepackage{colortbl}
\usepackage{theorem}
\usepackage{flushend}
\let\labelindent\relax
\usepackage{enumitem}

\newtheorem{theorem}{Theorem}

\newtheorem{example}[theorem]{Example}
\newtheorem{remark}[theorem]{Remark}
\newtheorem{definition}{Definition}

\usepackage[final,tracking=true,kerning=true,spacing=true,factor=1100,stretch=10,shrink=20]{microtype}

\IEEEoverridecommandlockouts                             

\title{\LARGE \bf
Twisted Gabidulin Codes in the GPT Cryptosystem
}

\author{Sven Puchinger, Julian Renner, Antonia Wachter-Zeh%
\thanks{S. Puchinger, J. Renner and A. Wachter-Zeh are with the Institute for Communications Engineering, Technical University of Munich, Germany.
Emails: \{sven.puchinger, julian.renner, antonia.wachter-zeh\}@tum.de}%
\thanks{This work was supported by the Technical University of Munich-Institute for Advanced Study, funded by the German Excellence Initiative and European Union Seventh Framework Programme under Grant Agreement No. 291763 and the German Research Foundation (Deutsche Forschungsgemeinschaft, DFG) unter Grant No. WA3907/1-1.}%
}%

\def\ve#1{{\mathchoice{\mbox{\boldmath$\displaystyle #1$}}%
		{\mbox{\boldmath$\textstyle #1$}}%
		{\mbox{\boldmath$\scriptstyle #1$}}%
		{\mbox{\boldmath$\scriptscriptstyle #1$}}}}

\renewcommand{\c}{\ve{c}}

\newcommand{\x}{\ve{x}}
\newcommand{\y}{\ve{y}}

\newcommand{\C}{\ve{C}}
\newcommand{\A}{\ve{A}}

\newcommand{\X}{\ve{X}}

\newcommand{\G}{\ve{G}}
\renewcommand{\H}{\ve{H}}

\newcommand{\NN}{\mathbb{N}}

\newcommand{\Fq}{\mathbb{F}_q}
\newcommand{\Fqm}{\mathbb{F}_{q^m}}

\newcommand{\Code}{\mathcal{C}}

\DeclareMathOperator{\rank}{rank}
\DeclareMathOperator{\dR}{d_R}

\newcommand{\evpolys}{\mathcal{P}^{n,k}_{\tVec,\hVec,\etaVec}}
\newcommand{\twisted}{$[k,\tVec,\hVec,\etaVec]$-twisted }
\newcommand{\twistedC}{$[\alphaVec,\tVec,\hVec,\etaVec]$-twisted }
\newcommand{\Cmult}{\TRS{\alphaVec,\tVec,\hVec,\etaVec}{n,k}{}}
\newcommand{\numTwists}{\ell}

\newcommand{\tVec}{{\ve t}}
\newcommand{\hVec}{{\ve h}}
\newcommand{\etaVec}{{\ve \eta}}

\newcommand{\Lset}{\mathcal{L}_{q^m}}
\newcommand{\qdeg}{\deg_q}

\newcommand{\alphaVec}{{\ve \alpha}}
\newcommand{\betaVec}{{\ve \beta}}

\newcommand{\TRS}[3]{\Code_{#1}^{#3}[#2]}
\newcommand{\LambdaOp}{\Lambda}

\newcommand{\qpow}[1]{^{[#1]}}

\newcommand{\Fqsi}[1]{\mathbb{F}_{q^{s_{#1}}}}

\newcommand{\pcVec}{\ve{\gamma}}
\newcommand{\pc}{\gamma}

\begin{document}

\maketitle
\thispagestyle{empty}
\pagestyle{empty}

\begin{abstract}
In this paper, we investigate twisted Gabidulin codes in the GPT code-based public-key cryptosystem.
We show that Overbeck's attack is not feasible for a subfamily of twisted Gabidulin codes.
The resulting key sizes are significantly lower than in the original McEliece system and also slightly smaller than in Loidreau's unbroken GPT variant.
\end{abstract}

\section{Introduction}

A \emph{rank-metric code} is a set of matrices whose distances are measured by the rank of their difference.
These codes, as well as their most famous family, \emph{Gabidulin codes}, were independently introduced in \cite{Delsarte_1978,Gabidulin_TheoryOfCodes_1985,Roth_RankCodes_1991}.
Gabidulin codes are maximum rank distance (MRD) codes, i.e., they fulfill the rank-metric Singleton bound with equality.

\emph{Twisted Gabidulin codes} were introduced by Sheekey in~\cite{sheekey2016new}, and are defined by adding an extra monomial to the evaluation polynomials of a Gabidulin code and choosing its coefficient in a suitable way such that the new codes remain MRD.
A special case of Sheekey's twisted Gabidulin codes was independently introduced in \cite{otal2016explicit}.
The idea of twisted Gabidulin codes was adapted to the Hamming metric and generalized in \cite{beelen2017twisted,beelen2018structural}, resulting in so-called twisted Reed--Solomon (RS) codes.
The methods developed for the latter codes were used to widely generalize Sheekey's construction in~\cite{puchinger2017further}.

The \emph{Gabidulin--Paramonov--Tretjakov (GPT) cryptosystem}~\cite{gabidulin1991ideals} is a modification of the McEliece code-based public-key cryptosystem using rank-metric codes.
Its motivation arises from the fact that all known generic rank-metric decoders \cite{chabaud1996cryptographic,ourivski2002new,gaborit2016complexity,aragon2017improved} are exponential in the square of the code parameters, which, compared to codes in Hamming metric, results in much smaller key sizes for a given security level (in cases where generic decoding is the most efficient attack).
The original GPT system was modified several times \cite{gabidulin2001modified,gabidulin2003reducible,loidreau2010designing,rashwan2011security,gabidulin2008attacks,gabidulin2009improving,rashwan2010smart} due to efficient attacks by Gibson \cite{gibson1995severely,gibson1996security} and Overbeck \cite{overbeck2006extending,overbeck2005new,overbeck2008structural}. However, most of these systems were broken by variants of Overbeck's attack, cf.~\cite{horlemann2016considerations,otmani2016improved,horlemann2018extension}.
The only variants that have not been broken so far are the one by Loidreau \cite{loidreau2016evolution,loidreau2017new} and the one by Berger et al.~\cite{berger2017gabidulin}.

Recently, it was shown that a subfamily of twisted RS codes resists several known structural attacks on the McEliece cryptosystem based on RS-like codes \cite{beelen2018structural}.
In this paper, we establish an analogous result in the rank metric.
We prove that certain twisted Gabidulin codes resist Overbeck's attack which provides key sizes that are smaller than comparable code-based systems for the same security levels.

\section{Preliminaries}

Let $\Fqm$ and $\Fq$ be finite fields, where $\Fqm$ is an extension field of $\Fq$.
It is well-known that $\Fqm$ is also an $m$-dimensional vector space over $\Fq$ and that elements of $\Fqm$ can be uniquely represented in a fixed basis of $\Fqm$ over $\Fq$.
Hence, we can represent a vector $\c \in \Fqm^n$ of length $n$ as an $m \times n$ matrix $\C \in \Fq^{m \times n}$ by expanding the entries of $\c$ column-wise.
The rank $\rank(\c)$ of a vector $\c \in \Fqm$ is defined by the rank of its matrix representation.

\subsection{Rank-Metric Codes}

The \emph{rank metric} is defined by
\begin{equation*}
\dR \, : \, \Fqm \times \Fqm \to \NN_0, \quad (\x,\y) \mapsto \rank(\x-\y).
\end{equation*}
It can be shown that the rank metric is indeed a metric. A linear rank-metric code of length $n$, dimension $k$ and minimum (rank) distance $d$ over $\Fqm$, denoted by $\Code$ or $\Code[n,k]$, is a $k$-dimensional subspace of $\Fqm^n$, which fulfills
\begin{equation*}
d = \min\limits_{\substack{\c_1,\c_2 \in \Code \\ \c_1\neq\c_2}} \dR(\c_1,\c_2).
\end{equation*}
The rank-metric Singleton bound states that the minimum distance of such a code with $n \leq m$ is upper-bounded by
$d \leq n-k+1$.
If a code fulfills this upper bound with equality, then it is called \emph{maximum rank distance (MRD)} code.

\subsection{Linearized Polynomials}

Linearized polynomials were first studied in \cite{ore1933special}. They are polynomials of the form
$f = \sum_{i} f_i x^{q^i}$,
where $f_i \in \Fqm$ and $f_i \neq 0$ for finitely many $i$.
For notational convenience, we define $[i] := q^i$, so we write $f = \sum_{i} f_i x^{[i]}$.
The $q$-degree of $f$ is defined by
\begin{equation*}
\qdeg f := \begin{cases}
\max\{i \, : \, f_i \neq 0\}, &\text{if } f \neq 0, \\
-\infty, &\text{if } f = 0.
\end{cases}
\end{equation*}
The evaluation map
$f(\cdot) \, : \, \Fqm \to \Fqm$, $\alpha \mapsto \sum_{i} f_i \alpha^{q^i}$
is $\Fq$-linear.
Using ordinary addition and composition of polynomials as multiplication, linearized polynomials form a (non-commutative) ring, which we denote by $\Lset$.

\subsection{Gabidulin Codes}

Gabidulin codes are MRD codes that were independently introduced in \cite{Delsarte_1978,Gabidulin_TheoryOfCodes_1985,Roth_RankCodes_1991}. They are defined as follows.
\begin{definition}[Gabidulin code\cite{Delsarte_1978,Gabidulin_TheoryOfCodes_1985,Roth_RankCodes_1991}]
	Let $\alpha_1,\dots,\alpha_n \in \Fqm$ be linearly independent over $\Fq$ and $k<n$. The corresponding $[n,k]$ Gabidulin code is defined by
	\begin{equation*}
	\Code_\mathrm{Gab}\! = \Big\{\! \big[ f(\alpha_1), f(\alpha_2), \dots, f(\alpha_n)\big] : f \in \Lset, \, \qdeg f\! < k \Big\}.
	\end{equation*}
\end{definition}

\subsection{Sum Operator}

\begin{definition}[$q$-Sum]
Let $\Code[n,k]$ be a linear code over $\Fqm$ and $i \in \NN_0$. Then, the ($i^\mathrm{th}$) $q$-sum of $\Code$ is defined by
\begin{equation*}
\LambdaOp_i(\Code) = \Code + \Code\qpow{1} + \dots + \Code\qpow{i}.
\end{equation*}
\end{definition}

It is well-known that a random code fulfills
\begin{equation*}
\dim \LambdaOp_i(\Code) = \min\{n,i k\}
\end{equation*}
with high probability and that a Gabidulin code satisfies
\begin{equation*}
\dim \LambdaOp_i(\Code) = \min\{n,k+i\}.
\end{equation*}
Hence, for $k<n-i$, a Gabidulin code has smaller $q$-sum dimension than a random code with high probability.
This observation was used in \cite{overbeck2005new,overbeck2008structural} to obtain a distinguisher and an attack on the GPT cryptosystem.

We will also use the $q$-sum operator for matrices.
\begin{definition}[$q$-Sum]
Let $\A \in \Fq^{s \times n}$ and $i \in \NN_0$. Then, the ($i^\mathrm{th}$) $q$-sum of $\A$ is defined by
\begin{equation*}
\LambdaOp_i(\A) = 
\begin{bmatrix}
	\A\\
	\A^{[1]}\\
	\vdots\\
	\A^{[i]}
\end{bmatrix}
\end{equation*}
\end{definition}
If $\G$ is a generator matrix of a code $\Code$, then $\LambdaOp_i(\G)$ is a generator matrix of $\LambdaOp_i(\Code)$.

\section{Variants of the GPT Cryptosystem and Overbeck's Attack}

\subsection{The GPT Crpytosystem and Its Variants}

The GPT cryptosystem is an instantiation of the McEliece cryptosystem with Gabidulin codes.
In this paper, we study the most general form of the GPT cryptosystem, cf.~\cite{Overbeck-ExtendingGibsonAttack}.

\textbf{Key Generation:}
Let $\G$ be the $k \times n$ generator matrix of an $[n,k]$ Gabidulin code,  $\ve{S}$ a random full-rank $k \times k$ matrix over $\Fqm$, $\X$ a random matrix of size $k \times \lambda$ over $\Fqm$ of rank $1 \leq s \leq \lambda$
and $\ve{P}$ a random $(n+\lambda) \times (n+\lambda)$ full-rank matrix  over $\Fq$. Then, the public key of the GPT cryptosystem is defined as:
\begin{equation}\label{eq:gpt_pubkey}
\G_{pub} = \ve{S} [\ve{X} | \G]\ve{P}
\end{equation}
and the integers $n$, $\lambda$, $k$, $t = \left\lfloor\frac{n-k}{2} \right \rfloor$.

To ensure proper decoding, it is important that $\ve{P}$ lies in $\Fq$. However, this enables Overbeck's attack (see the following subsection). In \cite{gabidulin2009improving,rashwan2011security}, variants for $\ve{P}$ were suggested where $\ve{P}$ is in $\Fqm$. However, as shown in \cite{otmani2016improved}, in all of these variants, the public key can be rewritten as
\begin{equation*}
\G_{pub} = \ve{S}^* [\ve{X}^* | \ve{G}^*] \ve{P}^*,
\end{equation*}
where $\ve{P}^*$ is in $\Fq$.

\textbf{Encryption:} To encode a plaintext $\ve{m}$, choose randomly a vector $\ve{z} \in \Fqm^n$ of rank $t$ and compute the ciphertext as
\begin{equation*}
\ve{c} = \ve{m}\G_{pub} + \ve{z}.
\end{equation*}

\textbf{Decryption:}
Apply the decoding algorithm corresponding to $\G$ to the last $n$ symbols of $\ve{c}\ve{P}^{-1}$. Clearly $\rank (\ve{z}\ve{P}^{-1}) \leq t$. 
This decoder therefore provides $\ve{m}\ve{S}$ and by inverting $\ve{S}$, the secret message $\ve{m}$ can be recovered.

\subsection{Overbeck's Attack}\label{ssec:Overbecks_Attack}

Overbeck's attack \cite{overbeck2005new,overbeck2008structural} is based on two observations:

\begin{enumerate}[label=\roman*)]
	\item\label{itm:h_in_dual_space} An $[n,k]$ Gabidulin code has a parity-check matrix of the form
	\begin{equation*}
	\!\!\!\!\!\!\!\!\!\!\H \!=\! \begin{bmatrix}
	\pcVec \\
	\pcVec\qpow{1} \\
	\vdots \\
	\pcVec\qpow{n-k-1} \\
	\end{bmatrix} 
	\!\!:=\! \begin{bmatrix}
	\pc_1 & \pc_2 & \dots & \pc_n \\
	\pc_1\qpow{1} & \pc_2\qpow{1} & \dots & \pc_n\qpow{1} \\
	\vdots & \vdots & \ddots & \vdots \\
	\pc_1\qpow{n-k-1}\! & \pc_2\qpow{n-k-1}\! & \dots &\pc_n\qpow{n-k-1}\!
	\end{bmatrix}\!,
	\end{equation*}
	where $\pcVec \in \Fqm^n$ with linearly independent entries.
	\item\label{itm:Lambda_i_n-1} There is an integer $i$ such that
	\begin{equation*}
	\dim \LambdaOp_i(\Code) = n-1.
	\end{equation*}
\end{enumerate}
The idea is to recover a vector $\pcVec$ by choosing a non-zero vector in the dual code of $\LambdaOp_i(\Code)$.

In the following, we show that property (ii) is fulfilled for the most general form of the GPT system, the one with a public key as in~\eqref{eq:gpt_pubkey}.

The generator matrix of the code $\LambdaOp_i(\Code_{pub})$ is:
\begin{equation*}
\LambdaOp_i(\G_{pub}) = 
\begin{bmatrix}
	\G_{pub}\\
	\G_{pub}^{[1]}\\
	\vdots\\
	\G_{pub}^{[i]}
\end{bmatrix}
 =  
 \ve{S}^\prime \cdot
 \begin{bmatrix}
 	\ve{X} | \G\\
\ve{X}^{[1]} | \G^{[1]}\\
\vdots\\
\ve{X}^{[i]} | \G^{[i]}\\
 \end{bmatrix} \cdot \ve{P},
\end{equation*}
where $\ve{P}^{[i]} = \ve{P}$ since $\ve{P} \in \Fq^{\ell + n}$ and $\ve{S}^\prime$ is a block diagonal matrix with $\ve{S}, \ve{S}^{[1]}, \dots, \ve{S}^{[i]}$ on the diagonal.

Since
\begin{equation*}
\rank
\begin{bmatrix}
\ve{X}\\
\ve{X}^{[1]}\\
\vdots\\
\ve{X}^{[i]}\\
\end{bmatrix} 
\leq \min\{\lambda, i\},
\
\text{and}
\
\rank
\begin{bmatrix}
\G\\
\G^{[1]}\\
\vdots\\
\G^{[i]}\\
\end{bmatrix} 
 = k +i,
\end{equation*}
we have $\dim \LambdaOp_i(\Code_{pub}) \leq \min\{n, k+i+t\}$.
For $i = n-k-t-1$, we get that $\rank \LambdaOp_i(\G_{pub}) = n-1$ and there exists a vector $\ve{v} = (\ve{0} |\ve{v}^\prime)$ such that $\LambdaOp_i(\G_{pub}) \cdot \ve{v}^T = \ve{0}$ and Overbeck's attack can be applied by simply using a polynomial-time decoder on the public code.

\section{Twisted Gabidulin Codes}

\subsection{Definition}

\begin{definition}[Twisted Gabidulin Code, \cite{puchinger2017further}]\label{def:tgab_definition}
Let $n,k,\numTwists \in \NN$ with $k < n$ and $\ell \leq n-k$. Choose a 
\begin{itemize}
\item \emph{hook vector}\footnote{For didactic reasons, this definition slightly differs from the one in \cite{puchinger2017further}, i.e., is a special case.} $\hVec \in \{0,\dots,k-1\}^\numTwists$ and a
\item \emph{twist vector} $\tVec \in \{1,\dots,n-k\}^\numTwists$ with distinct $t_i$, and let
\item $\etaVec \in (\Fqm \setminus \{0\})^\numTwists$.
\end{itemize}
The set of \emph{\twisted linearized polynomials} over $\Fqm$ is defined by
\begin{equation*}
\evpolys = \left\{ f = \sum_{i=0}^{k-1} f_i x\qpow{i} + \sum_{j=1}^{\numTwists} \eta_j f_{h_j} x\qpow{k-1+t_j} : f_i \in \Fqm \right\}.
\end{equation*}
Let $\alpha_1,\dots,\alpha_n \in \Fqm$ be linearly independent over $\Fq$ and write $\alphaVec = [\alpha_1,\dots,\alpha_n]$.
The \emph{\twistedC Gabidulin code} of length $n$ and dimension $k$ is given by
\begin{equation*}
\Cmult = \left\{ \big[f(\alpha_1),f(\alpha_2), \dots, f(\alpha_n)] \, : \, f \in \evpolys \right\}.
\end{equation*}
\end{definition}

\subsection{Generator Matrix of a Twisted Gabiulin Code}

A generator matrix of a twisted Gabidulin code with $h_1 < h_2 < \dots < h_\ell$ is given by:

\begin{equation*}
\G_\mathrm{TGab} =
\begin{bmatrix}
\alphaVec \\
\alphaVec\qpow{1} \\
\vdots \\
\alphaVec\qpow{h_1-1} \\
\alphaVec\qpow{h_1} + \eta_1 \alphaVec\qpow{k-1+t_1} \\
\alphaVec\qpow{h_1+1} \\
\vdots \\
\alphaVec\qpow{h_\ell-1} \\
\alphaVec\qpow{h_\ell} + \eta_\ell \alphaVec\qpow{k-1+t_\ell} \\
\alphaVec\qpow{h_\ell+1} \\
\vdots \\
\alphaVec\qpow{k-1}
\end{bmatrix}.
\end{equation*}

\begin{example}\label{ex:Example_Generator_Matrix}
	The following matrix is a example for the generator matrix of an $[n,k] = [27,10]$ twisted Gabidulin code with $\tVec = [6,12]$, $\hVec = [5,6]$, see also the illustration in Figure~\ref{fig:example_G}:
\begin{align*}
\G_\mathrm{TGab} =
\begin{bmatrix}
\alphaVec \\
\alphaVec\qpow{1} \\
\alphaVec\qpow{2} \\
\alphaVec\qpow{3} \\
\alphaVec\qpow{4} \\
\alphaVec\qpow{5} + \eta_1 \alphaVec\qpow{15} \\
\alphaVec\qpow{6} + \eta_2 \alphaVec\qpow{21} \\
\alphaVec\qpow{7} \\
\alphaVec\qpow{8} \\
\alphaVec\qpow{9} \\
\end{bmatrix}
\end{align*}
\end{example}

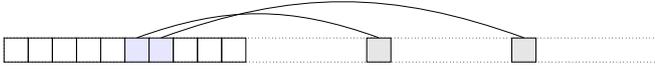
\begin{figure}[ht]
\begin{center}
	\resizebox{\columnwidth}{!}{
	\begin{tikzpicture}
	\def\ylabelpos{-1cm}
	\def\xlevelminusone{-3cm}
	\def\xlevelzero{0cm}
	\def\xlevelone{3cm}
	\def\xleveltwo{6cm}
	\def\xlevelthree{9cm}
	\def\xwidth{0.6cm}
	\def\ywidth{0.6cm}
	\def\ydist{1cm}
	\def\nval{27}
	\def\kval{5}
	\def\twistval{6}
	\def\twistvaltwo{8}
	\def\twistvalthree{12}
	\def\myeps{0.1cm}
	\def\yshift{-5cm}
	\tikzstyle{coeffNodesPure}=[draw, rectangle, color=blue, minimum width=\xwidth, minimum height=\ywidth, inner sep=0pt]
	\tikzstyle{coeffNodes}=[draw, rectangle, minimum width=\xwidth, minimum height=\ywidth, inner sep=0pt]
	\tikzstyle{hoookNodes}=[draw, rectangle, fill=blue!10, minimum width=\xwidth, minimum height=\ywidth, inner sep=0pt]
	\tikzstyle{twistNodes}=[draw, rectangle, fill=black!10, minimum width=\xwidth, minimum height=\ywidth, inner sep=0pt]
	\tikzstyle{coeffNodesLevels}=[draw, rectangle, color=blue, minimum width=\xwidth, minimum height=\ywidth, inner sep=0pt]

	\draw[dotted] (-0.5*\xwidth,-0.5*\ywidth) rectangle (\nval*\xwidth-0.5*\xwidth,0.5*\ywidth);
	\node[coeffNodes]  (levelOnecoeff0) at (0*\xwidth,0) {};
	\node[coeffNodes]  (levelOnecoeff1) at (1*\xwidth,0) {};
	\node[coeffNodes]  (levelOnecoeff2) at (2*\xwidth,0) {};
	\node[coeffNodes]  (levelOnecoeff3) at (3*\xwidth,0) {};
	\node[coeffNodes]  (levelOnecoeff4) at (4*\xwidth,0) {};
	\node[hoookNodes]  (levelOnecoeff5) at (5*\xwidth,0) {};
	\node[hoookNodes]  (levelOnecoeff6) at (6*\xwidth,0) {};
	\node[coeffNodes]  (levelOnecoeff7) at (7*\xwidth,0) {};
	\node[coeffNodes]  (levelOnecoeff8) at (8*\xwidth,0) {};
	\node[coeffNodes]  (levelOnecoeff9) at (9*\xwidth,0) {};
	\node[twistNodes] (levelOnetwist1) at (15*\xwidth,0) {};
	\node[twistNodes] (levelOnetwist2) at (21*\xwidth,0) {};
	\draw[bend angle=20, bend left] (levelOnecoeff5.north) to (levelOnetwist1.north);
	\draw[bend angle=20, bend left] (levelOnecoeff6.north) to (levelOnetwist2.north);
	
	\end{tikzpicture}
	}
\end{center}
\caption{Illustration of the {generator matrix} in Example~\ref{ex:Example_Generator_Matrix}.}
\label{fig:example_G}
\end{figure}

\subsection{MRD Twisted Gabidulin Codes}

It was shown in \cite{puchinger2017further} that if the $\eta_j$ and the evaluation points $\alpha_i$ are chosen in a proper way, the twisted Gabidulin code $\Cmult$ is MRD. 
\begin{theorem}[Twisted Gabidulin are MRD, {\cite[Thm~1]{puchinger2017further}}]\label{thm:MRD_condition}
Let $s_0,\dots,s_\ell \in \NN$ be such that
\begin{equation*}
\Fq \subseteq \Fqsi{0} \subsetneq \Fqsi{1} \subsetneq \dots \subsetneq \Fqsi{\ell} = \Fqm
\end{equation*}
is a chain of subfields. Choose $k<n\leq s_0$ and $\hVec,\tVec,\etaVec,\alphaVec$ as in Definition~\ref{def:tgab_definition} with the additional requirements
\begin{align*}
\alpha_1,\dots,\alpha_n &\in \Fqsi{0}, \text{ and} \\
\eta_i &\in \Fqsi{i} \setminus \Fqsi{i-1} \quad \forall \, i=1,\dots,\ell.
\end{align*}
Then, $\Cmult$ is MRD.
\end{theorem}

\begin{remark}
We can choose $s_i = 2 \cdot s_{i-1}$ in Theorem~\ref{thm:MRD_condition}. In this way, we obtain
\begin{equation*}
m = 2^\ell s_0 \geq 2^\ell n.
\end{equation*}
This means that when $n=s_0$, the memory required to store a generator matrix of a twisted Gabidulin code with $\ell$ twists is $2^\ell$ times larger than the one of a Gabidulin code with the same parameters $[n,k]$.
\end{remark}

\subsection{Decoding}

Finding an efficient decoding algorithm for twisted Gabidulin codes is an open problem and research in progress.
For the original twisted Gabidulin codes by Sheekey~\cite{sheekey2016new}, an efficient decoding algorithm was found in \cite{rosenthal2017decoding,randrianarisoa2017decoding}.
Although a generalization to the twisted codes in \cite{puchinger2017further} does not appear to be straightforward, it seems likely that an efficient decoder will be found soon.

\section{Twisted Gabidulin Codes With Large $q$-Sum Dimension}

In this section, we show that twisted Gabidulin codes can have larger $q$-sum dimension than a Gabidulin code of the same dimension would have.

\begin{theorem}[Large $q$-Sum Dimension]\label{thm:large_q-sum-dim_family}
	Let $n,k,\tVec,\hVec,\etaVec,\alphaVec$ be chosen as in Definition~\ref{def:tgab_definition} such that
	\begin{align}
	&\Delta := \tfrac{n-k-\ell}{\ell+1} \in \NN, \label{eq:cond_Delta_integer} \\
	&t_i := (i+1)(\Delta+1), &&\forall \, i =1,\dots,\ell, \label{eq:cond_tVec}\\
	&0 < h_1 < h_2 < \dots h_\ell < k-1 &&\text{s.t.} \notag\\
	&|h_{i+1}-h_i| > 1, &&\forall \, i = 1,\dots,\ell-1. \label{eq:cond_hVec}
	\end{align}
	For all $i \in \NN$, we then have
	\begin{equation*}
	\dim \LambdaOp_i(\Cmult) = \min\{k-1+(i+1)(\ell+1),n\}.
	\end{equation*}
\end{theorem}

\begin{proof}
Let $\Code := \Cmult$.
We first prove that
\begin{align*}
\G_{\LambdaOp_1(\Code)} :=
\begin{bmatrix}
\alphaVec\qpow{0} \\
\alphaVec\qpow{1} \\
\vdots \\
\alphaVec\qpow{k} \\
\alphaVec\qpow{k-1+t_1} \\
\alphaVec\qpow{k+t_1} \\
\alphaVec\qpow{k-1+t_2} \\
\alphaVec\qpow{k+t_2} \\
\vdots \\
\alphaVec\qpow{k-1+t_\ell} \\
\alphaVec\qpow{k+t_\ell}
\end{bmatrix}
\end{align*}
is a generator matrix of the code $\LambdaOp_1(\Code)$. The proof is illustrated in Figure~\ref{fig:Lambda_1_basis_statement}.

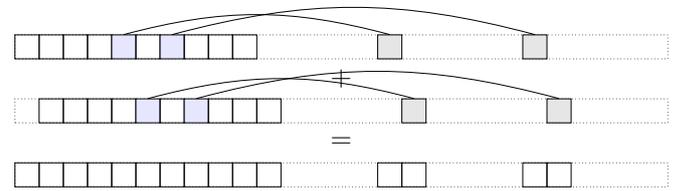
\begin{figure}[ht!]
	\begin{center}
		\resizebox{\columnwidth}{!}{
			\begin{tikzpicture}
			\def\ylabelpos{-1cm}
			\def\xlevelminusone{-3cm}
			\def\xlevelzero{0cm}
			\def\xlevelone{3cm}
			\def\xleveltwo{6cm}
			\def\xlevelthree{9cm}
			\def\xwidth{0.6cm}
			\def\ywidth{0.6cm}
			\def\ydist{1cm}
			\def\nval{27}
			\def\kval{5}
			\def\twistval{6}
			\def\twistvaltwo{8}
			\def\twistvalthree{12}
			\def\myeps{0.1cm}
			\def\yshift{-1.6cm}
			\tikzstyle{coeffNodesPure}=[draw, rectangle, color=blue, minimum width=\xwidth, minimum height=\ywidth, inner sep=0pt]
			\tikzstyle{coeffNodes}=[draw, rectangle, minimum width=\xwidth, minimum height=\ywidth, inner sep=0pt]
			\tikzstyle{hoookNodes}=[draw, rectangle, fill=blue!10, minimum width=\xwidth, minimum height=\ywidth, inner sep=0pt]
			\tikzstyle{twistNodes}=[draw, rectangle, fill=black!10, minimum width=\xwidth, minimum height=\ywidth, inner sep=0pt]
			\tikzstyle{coeffNodesLevels}=[draw, rectangle, color=blue, minimum width=\xwidth, minimum height=\ywidth, inner sep=0pt]

			\draw[dotted] (-0.5*\xwidth,-0.5*\ywidth) rectangle (\nval*\xwidth-0.5*\xwidth,0.5*\ywidth);
			\node[coeffNodes]  (levelOnecoeff0) at (0*\xwidth,0) {};
			\node[coeffNodes]  (levelOnecoeff1) at (1*\xwidth,0) {};
			\node[coeffNodes]  (levelOnecoeff2) at (2*\xwidth,0) {};
			\node[coeffNodes]  (levelOnecoeff3) at (3*\xwidth,0) {};
			\node[hoookNodes]  (levelOnecoeff4) at (4*\xwidth,0) {};
			\node[coeffNodes]  (levelOnecoeff5) at (5*\xwidth,0) {};
			\node[hoookNodes]  (levelOnecoeff6) at (6*\xwidth,0) {};
			\node[coeffNodes]  (levelOnecoeff7) at (7*\xwidth,0) {};
			\node[coeffNodes]  (levelOnecoeff8) at (8*\xwidth,0) {};
			\node[coeffNodes]  (levelOnecoeff9) at (9*\xwidth,0) {};
			\node[twistNodes] (levelOnetwist1) at (15*\xwidth,0) {};
			\node[twistNodes] (levelOnetwist2) at (21*\xwidth,0) {};
			\draw[bend angle=15, bend left] (levelOnecoeff4.north) to (levelOnetwist1.north);
			\draw[bend angle=15, bend left] (levelOnecoeff6.north) to (levelOnetwist2.north);
			
			\draw[dotted] (-0.5*\xwidth,\yshift-0.5*\ywidth) rectangle (\nval*\xwidth-0.5*\xwidth,\yshift+0.5*\ywidth);
			\node[coeffNodes]  (levelTwocoeff0) at (1*\xwidth,\yshift) {};
			\node[coeffNodes]  (levelTwocoeff1) at (2*\xwidth,\yshift) {};
			\node[coeffNodes]  (levelTwocoeff2) at (3*\xwidth,\yshift) {};
			\node[coeffNodes]  (levelTwocoeff3) at (4*\xwidth,\yshift) {};
			\node[hoookNodes]  (levelTwocoeff4) at (5*\xwidth,\yshift) {};
			\node[coeffNodes]  (levelTwocoeff5) at (6*\xwidth,\yshift) {};
			\node[hoookNodes]  (levelTwocoeff6) at (7*\xwidth,\yshift) {};
			\node[coeffNodes]  (levelTwocoeff7) at (8*\xwidth,\yshift) {};
			\node[coeffNodes]  (levelTwocoeff8) at (9*\xwidth,\yshift) {};
			\node[coeffNodes]  (levelTwocoeff9) at (10*\xwidth,\yshift) {};
			\node[twistNodes]  (levelTwotwist1) at (16*\xwidth,\yshift) {};
			\node[twistNodes]  (levelTwotwist2) at (22*\xwidth,\yshift) {};
			\draw[bend angle=15, bend left] (levelTwocoeff4.north) to (levelTwotwist1.north);
			\draw[bend angle=15, bend left] (levelTwocoeff6.north) to (levelTwotwist2.north);
			
			\draw[dotted] (-0.5*\xwidth,2*\yshift-0.5*\ywidth) rectangle (\nval*\xwidth-0.5*\xwidth,2*\yshift+0.5*\ywidth);
			\node[coeffNodes]  (levelThreecoeff0) at (0*\xwidth,2*\yshift) {};
			\node[coeffNodes]  (levelThreecoeff0) at (1*\xwidth,2*\yshift) {};
			\node[coeffNodes]  (levelThreecoeff1) at (2*\xwidth,2*\yshift) {};
			\node[coeffNodes]  (levelThreecoeff2) at (3*\xwidth,2*\yshift) {};
			\node[coeffNodes]  (levelThreecoeff3) at (4*\xwidth,2*\yshift) {};
			\node[coeffNodes]  (levelThreecoeff4) at (5*\xwidth,2*\yshift) {};
			\node[coeffNodes]  (levelThreecoeff5) at (6*\xwidth,2*\yshift) {};
			\node[coeffNodes]  (levelThreecoeff6) at (7*\xwidth,2*\yshift) {};
			\node[coeffNodes]  (levelThreecoeff7) at (8*\xwidth,2*\yshift) {};
			\node[coeffNodes]  (levelThreecoeff8) at (9*\xwidth,2*\yshift) {};
			\node[coeffNodes]  (levelThreecoeff10) at (10*\xwidth,2*\yshift) {};
			\node[coeffNodes]  (levelThreecoeff11) at (15*\xwidth,2*\yshift) {};
			\node[coeffNodes]  (levelThreecoeff12) at (16*\xwidth,2*\yshift) {};
			\node[coeffNodes]  (levelThreecoeff13) at (21*\xwidth,2*\yshift) {};
			\node[coeffNodes]  (levelThreecoeff14) at (22*\xwidth,2*\yshift) {};

			\node (plusLabel) at (13*\xwidth,0.5*\yshift) {\huge $+$};
			\node (plusLabel) at (13*\xwidth,1.5*\yshift) {\huge $=$};

			\end{tikzpicture}
		}
	\end{center}
	\caption{Illustration of the first part of Theorem~\ref{thm:large_q-sum-dim_family}'s proof.}
	\label{fig:Lambda_1_basis_statement}
\end{figure}

By Condition \eqref{eq:cond_hVec}, for any $i=1,\dots,\ell$, the vectors $\alphaVec^{h_i-1}$ and $\alphaVec^{h_i+1}$ are in $\Code$, so
\begin{equation*}
(\alphaVec^{h_i-1})\qpow{1} = \alphaVec^{h_i}, \quad
(\alphaVec^{h_i+1})\qpow{1} = \alphaVec^{h_i+2}
\end{equation*}
are in $\Code\qpow{1}$. Hence, for any $\alphaVec^{j}$ with $j=0,\dots,k$, $\alphaVec^{j}$ is in the code $\Code$ (if we have $j \neq h_i$ for all $i$) or in $\Code\qpow{1}$ (if, e.g., $j=h_i$ for some $i$).

Any $\alphaVec\qpow{k-1+t_i}$ is in $\LambdaOp_1(\Code)$ since
\begin{equation*}
\alphaVec\qpow{k-1+t_i} = \eta_i^{-1} \big(\underset{\in \, \Code}{\underbrace{\eta_i \alphaVec\qpow{k-1+t_i}+\alphaVec\qpow{h_i}}} \big) - \eta_i^{-1} \underset{\in \, \LambdaOp_1(\Code)}{\underbrace{\alphaVec\qpow{h_i}}},
\end{equation*}
and similarly, we have
\begin{equation*}
\alphaVec\qpow{k+t_i} = \big(\eta_i^{-1}\big)\qpow{-1} \underset{\in \, \Code\qpow{1}}{\underbrace{\big(\eta_i \alphaVec\qpow{k-1+t_i}+\alphaVec\qpow{h_i}\big)\qpow{1}}}  - \big(\eta_i^{-1}\big)\qpow{-1} \underset{\in \, \LambdaOp_1(\Code)}{\underbrace{\alphaVec\qpow{h_i}}}.
\end{equation*}

Since the $\alphaVec\qpow{i}$ are linearly independent for $0 \leq i <n$ and there are no other possible powers of $\alphaVec$ achievable by linear combinations of elements in $\Code$ and $\Code\qpow{1}$, the rows of the generator matrix $\G_{\LambdaOp_1(\Code)}$ are a basis of $\LambdaOp_1(\Code)$.
As the matrix has
$k+1+ 2 \ell = k-1+2(\ell+1)$
rows.

By the choice of the $t_i$, the ``power gaps'' between $\alphaVec\qpow{k}$ and $\alphaVec\qpow{k-1+t_1}$, as well as between $\alphaVec\qpow{k+t_i}$ and $\alphaVec\qpow{k-1+t_{i+1}}$ for any $i=1,\dots,\ell-1$, and between $\alphaVec\qpow{k+t_\ell}$ and $\alphaVec\qpow{n}$, are exactly $\Delta-1$.

If we iteratively increase the $q$-power of $\Lambda_i(\Code)$ to $\Lambda_{i+1}(\Code)$, the new basis elements are
\begin{equation*}
\alphaVec\qpow{k+i}, \, \alphaVec\qpow{k+t_1+i}, \, \alphaVec\qpow{k+t_2+i}, \dots, \, \alphaVec\qpow{k+t_\ell+i},
\end{equation*}
which are $\ell+1$ many.
This process can be repeated until we have $i=\Delta$, in which case the code $\Lambda_i(\Code)$ already contains all $\alphaVec\qpow{0},\alphaVec\qpow{1},\dots,\alphaVec\qpow{n-1}$ and has dimension $n$.

This implies the claim.
The iteration is illustrated in Figure~\ref{fig:Lambda_i_basis_iteration}.
\end{proof}

\begin{figure}[ht!]
$i=1$:
	\begin{center}
		\resizebox{\columnwidth}{!}{
			\begin{tikzpicture}
			\def\ylabelpos{-1cm}
			\def\xlevelminusone{-3cm}
			\def\xlevelzero{0cm}
			\def\xlevelone{3cm}
			\def\xleveltwo{6cm}
			\def\xlevelthree{9cm}
			\def\xwidth{0.6cm}
			\def\ywidth{0.6cm}
			\def\ydist{1cm}
			\def\nval{27}
			\def\kval{5}
			\def\twistval{6}
			\def\twistvaltwo{8}
			\def\twistvalthree{12}
			\def\myeps{0.1cm}
			\def\yshift{-1.3cm}
			\tikzstyle{coeffNodesPure}=[draw, rectangle, color=blue, minimum width=\xwidth, minimum height=\ywidth, inner sep=0pt]
			\tikzstyle{coeffNodes}=[draw, rectangle, minimum width=\xwidth, minimum height=\ywidth, inner sep=0pt]
			\tikzstyle{hoookNodes}=[draw, rectangle, fill=blue!10, minimum width=\xwidth, minimum height=\ywidth, inner sep=0pt]
			\tikzstyle{twistNodes}=[draw, rectangle, fill=black!10, minimum width=\xwidth, minimum height=\ywidth, inner sep=0pt]
			\tikzstyle{coeffNodesLevels}=[draw, rectangle, color=blue, minimum width=\xwidth, minimum height=\ywidth, inner sep=0pt]

			\draw[dotted] (-0.5*\xwidth,0*\yshift-0.5*\ywidth) rectangle (\nval*\xwidth-0.5*\xwidth,0*\yshift+0.5*\ywidth);
			\node[coeffNodes] at (0*\xwidth,0*\yshift) {};
			\node[coeffNodes] at (1*\xwidth,0*\yshift) {};
			\node[coeffNodes] at (2*\xwidth,0*\yshift) {};
			\node[coeffNodes] at (3*\xwidth,0*\yshift) {};
			\node[coeffNodes] at (4*\xwidth,0*\yshift) {};
			\node[coeffNodes] at (5*\xwidth,0*\yshift) {};
			\node[coeffNodes] at (6*\xwidth,0*\yshift) {};
			\node[coeffNodes] at (7*\xwidth,0*\yshift) {};
			\node[coeffNodes] at (8*\xwidth,0*\yshift) {};
			\node[coeffNodes] at (9*\xwidth,0*\yshift) {};
			\node[coeffNodes] at (10*\xwidth,0*\yshift) {};
			\node[coeffNodes] at (15*\xwidth,0*\yshift) {};
			\node[coeffNodes] at (16*\xwidth,0*\yshift) {};
			\node[coeffNodes] at (21*\xwidth,0*\yshift) {};
			\node[coeffNodes] at (22*\xwidth,0*\yshift) {};

			\draw[dotted] (-0.5*\xwidth,1*\yshift-0.5*\ywidth) rectangle (\nval*\xwidth-0.5*\xwidth,1*\yshift+0.5*\ywidth);
			\node[coeffNodes] at (1*\xwidth,1*\yshift) {};
			\node[coeffNodes] at (2*\xwidth,1*\yshift) {};
			\node[coeffNodes] at (3*\xwidth,1*\yshift) {};
			\node[coeffNodes] at (4*\xwidth,1*\yshift) {};
			\node[coeffNodes] at (5*\xwidth,1*\yshift) {};
			\node[coeffNodes] at (6*\xwidth,1*\yshift) {};
			\node[coeffNodes] at (7*\xwidth,1*\yshift) {};
			\node[coeffNodes] at (8*\xwidth,1*\yshift) {};
			\node[coeffNodes] at (9*\xwidth,1*\yshift) {};
			\node[coeffNodes] at (10*\xwidth,1*\yshift) {};
			\node[coeffNodes] at (11*\xwidth,1*\yshift) {};
			\node[coeffNodes] at (16*\xwidth,1*\yshift) {};
			\node[coeffNodes] at (17*\xwidth,1*\yshift) {};
			\node[coeffNodes] at (22*\xwidth,1*\yshift) {};
			\node[coeffNodes] at (23*\xwidth,1*\yshift) {};
			
			\draw[dotted] (-0.5*\xwidth,2*\yshift-0.5*\ywidth) rectangle (\nval*\xwidth-0.5*\xwidth,2*\yshift+0.5*\ywidth);
			\node[coeffNodes] at (0*\xwidth,2*\yshift) {};
			\node[coeffNodes] at (1*\xwidth,2*\yshift) {};
			\node[coeffNodes] at (2*\xwidth,2*\yshift) {};
			\node[coeffNodes] at (3*\xwidth,2*\yshift) {};
			\node[coeffNodes] at (4*\xwidth,2*\yshift) {};
			\node[coeffNodes] at (5*\xwidth,2*\yshift) {};
			\node[coeffNodes] at (6*\xwidth,2*\yshift) {};
			\node[coeffNodes] at (7*\xwidth,2*\yshift) {};
			\node[coeffNodes] at (8*\xwidth,2*\yshift) {};
			\node[coeffNodes] at (9*\xwidth,2*\yshift) {};
			\node[coeffNodes] at (10*\xwidth,2*\yshift) {};
			\node[coeffNodes] at (11*\xwidth,2*\yshift) {};
			\node[coeffNodes] at (15*\xwidth,2*\yshift) {};
			\node[coeffNodes] at (16*\xwidth,2*\yshift) {};
			\node[coeffNodes] at (17*\xwidth,2*\yshift) {};
			\node[coeffNodes] at (21*\xwidth,2*\yshift) {};
			\node[coeffNodes] at (22*\xwidth,2*\yshift) {};
			\node[coeffNodes] at (23*\xwidth,2*\yshift) {};
			
			\node (plusLabel) at (13*\xwidth,0.5*\yshift) {\huge $+$};
			\node (plusLabel) at (13*\xwidth,1.5*\yshift) {\huge $=$};

			\end{tikzpicture}
		}
	\end{center}
	\begin{center}
		$\vdots$
	\end{center}
	\vspace{-0.5cm}
	$i=4=\Delta-1$:
	\begin{center}
		\resizebox{\columnwidth}{!}{
			\begin{tikzpicture}
			\def\ylabelpos{-1cm}
			\def\xlevelminusone{-3cm}
			\def\xlevelzero{0cm}
			\def\xlevelone{3cm}
			\def\xleveltwo{6cm}
			\def\xlevelthree{9cm}
			\def\xwidth{0.6cm}
			\def\ywidth{0.6cm}
			\def\ydist{1cm}
			\def\nval{27}
			\def\kval{5}
			\def\twistval{6}
			\def\twistvaltwo{8}
			\def\twistvalthree{12}
			\def\myeps{0.1cm}
			\def\yshift{-1.3cm}
			\tikzstyle{coeffNodesPure}=[draw, rectangle, color=blue, minimum width=\xwidth, minimum height=\ywidth, inner sep=0pt]
			\tikzstyle{coeffNodes}=[draw, rectangle, minimum width=\xwidth, minimum height=\ywidth, inner sep=0pt]
			\tikzstyle{hoookNodes}=[draw, rectangle, fill=blue!10, minimum width=\xwidth, minimum height=\ywidth, inner sep=0pt]
			\tikzstyle{twistNodes}=[draw, rectangle, fill=black!10, minimum width=\xwidth, minimum height=\ywidth, inner sep=0pt]
			\tikzstyle{coeffNodesLevels}=[draw, rectangle, color=blue, minimum width=\xwidth, minimum height=\ywidth, inner sep=0pt]

			\draw[dotted] (-0.5*\xwidth,0*\yshift-0.5*\ywidth) rectangle (\nval*\xwidth-0.5*\xwidth,0*\yshift+0.5*\ywidth);
			\node[coeffNodes] at (0*\xwidth,0*\yshift) {};
			\node[coeffNodes] at (1*\xwidth,0*\yshift) {};
			\node[coeffNodes] at (2*\xwidth,0*\yshift) {};
			\node[coeffNodes] at (3*\xwidth,0*\yshift) {};
			\node[coeffNodes] at (4*\xwidth,0*\yshift) {};
			\node[coeffNodes] at (5*\xwidth,0*\yshift) {};
			\node[coeffNodes] at (6*\xwidth,0*\yshift) {};
			\node[coeffNodes] at (7*\xwidth,0*\yshift) {};
			\node[coeffNodes] at (8*\xwidth,0*\yshift) {};
			\node[coeffNodes] at (9*\xwidth,0*\yshift) {};
			\node[coeffNodes] at (10*\xwidth,0*\yshift) {};
			\node[coeffNodes] at (11*\xwidth,0*\yshift) {};
			\node[coeffNodes] at (12*\xwidth,0*\yshift) {};
			\node[coeffNodes] at (13*\xwidth,0*\yshift) {};
			
			\node[coeffNodes] at (15*\xwidth,0*\yshift) {};
			\node[coeffNodes] at (16*\xwidth,0*\yshift) {};
			\node[coeffNodes] at (17*\xwidth,0*\yshift) {};
			\node[coeffNodes] at (18*\xwidth,0*\yshift) {};
			\node[coeffNodes] at (19*\xwidth,0*\yshift) {};
			
			\node[coeffNodes] at (21*\xwidth,0*\yshift) {};
			\node[coeffNodes] at (22*\xwidth,0*\yshift) {};
			\node[coeffNodes] at (23*\xwidth,0*\yshift) {};
			\node[coeffNodes] at (24*\xwidth,0*\yshift) {};
			\node[coeffNodes] at (25*\xwidth,0*\yshift) {};

			\draw[dotted] (-0.5*\xwidth,1*\yshift-0.5*\ywidth) rectangle (\nval*\xwidth-0.5*\xwidth,1*\yshift+0.5*\ywidth);
			\node[coeffNodes] at (1*\xwidth,1*\yshift) {};
			\node[coeffNodes] at (2*\xwidth,1*\yshift) {};
			\node[coeffNodes] at (3*\xwidth,1*\yshift) {};
			\node[coeffNodes] at (4*\xwidth,1*\yshift) {};
			\node[coeffNodes] at (5*\xwidth,1*\yshift) {};
			\node[coeffNodes] at (6*\xwidth,1*\yshift) {};
			\node[coeffNodes] at (7*\xwidth,1*\yshift) {};
			\node[coeffNodes] at (8*\xwidth,1*\yshift) {};
			\node[coeffNodes] at (9*\xwidth,1*\yshift) {};
			\node[coeffNodes] at (10*\xwidth,1*\yshift) {};
			\node[coeffNodes] at (11*\xwidth,1*\yshift) {};
			\node[coeffNodes] at (12*\xwidth,1*\yshift) {};
			\node[coeffNodes] at (13*\xwidth,1*\yshift) {};
			\node[coeffNodes] at (14*\xwidth,1*\yshift) {};
			
			\node[coeffNodes] at (16*\xwidth,1*\yshift) {};
			\node[coeffNodes] at (17*\xwidth,1*\yshift) {};
			\node[coeffNodes] at (18*\xwidth,1*\yshift) {};
			\node[coeffNodes] at (19*\xwidth,1*\yshift) {};
			\node[coeffNodes] at (20*\xwidth,1*\yshift) {};
			
			\node[coeffNodes] at (22*\xwidth,1*\yshift) {};
			\node[coeffNodes] at (23*\xwidth,1*\yshift) {};
			\node[coeffNodes] at (24*\xwidth,1*\yshift) {};
			\node[coeffNodes] at (25*\xwidth,1*\yshift) {};
			\node[coeffNodes] at (26*\xwidth,1*\yshift) {};
			
			\draw[dotted] (-0.5*\xwidth,2*\yshift-0.5*\ywidth) rectangle (\nval*\xwidth-0.5*\xwidth,2*\yshift+0.5*\ywidth);
			\node[coeffNodes] at (0*\xwidth,2*\yshift) {};
			\node[coeffNodes] at (1*\xwidth,2*\yshift) {};
			\node[coeffNodes] at (2*\xwidth,2*\yshift) {};
			\node[coeffNodes] at (3*\xwidth,2*\yshift) {};
			\node[coeffNodes] at (4*\xwidth,2*\yshift) {};
			\node[coeffNodes] at (5*\xwidth,2*\yshift) {};
			\node[coeffNodes] at (6*\xwidth,2*\yshift) {};
			\node[coeffNodes] at (7*\xwidth,2*\yshift) {};
			\node[coeffNodes] at (8*\xwidth,2*\yshift) {};
			\node[coeffNodes] at (9*\xwidth,2*\yshift) {};
			\node[coeffNodes] at (10*\xwidth,2*\yshift) {};
			\node[coeffNodes] at (11*\xwidth,2*\yshift) {};
			\node[coeffNodes] at (12*\xwidth,2*\yshift) {};
			\node[coeffNodes] at (13*\xwidth,2*\yshift) {};
			\node[coeffNodes] at (14*\xwidth,2*\yshift) {};
			\node[coeffNodes] at (15*\xwidth,2*\yshift) {};
			\node[coeffNodes] at (16*\xwidth,2*\yshift) {};
			\node[coeffNodes] at (17*\xwidth,2*\yshift) {};
			\node[coeffNodes] at (18*\xwidth,2*\yshift) {};
			\node[coeffNodes] at (19*\xwidth,2*\yshift) {};
			\node[coeffNodes] at (20*\xwidth,2*\yshift) {};
			\node[coeffNodes] at (21*\xwidth,2*\yshift) {};
			\node[coeffNodes] at (22*\xwidth,2*\yshift) {};
			\node[coeffNodes] at (23*\xwidth,2*\yshift) {};
			\node[coeffNodes] at (24*\xwidth,2*\yshift) {};
			\node[coeffNodes] at (25*\xwidth,2*\yshift) {};
			\node[coeffNodes] at (26*\xwidth,2*\yshift) {};
			
			\node (plusLabel) at (13*\xwidth,0.5*\yshift) {\huge $+$};
			\node (plusLabel) at (13*\xwidth,1.5*\yshift) {\huge $=$};

			\end{tikzpicture}
		}
	\end{center}
	\caption{Illustration of the inductive argument in the second half of Theorem~\ref{thm:large_q-sum-dim_family}'s proof.}
	\label{fig:Lambda_i_basis_iteration}
\end{figure}

\section{Overbeck's Attack on Twisted Gabidulin Codes}

We consider the code family in Theorem~\ref{thm:large_q-sum-dim_family} and show that it is resistant against Overbeck's attack by considering its two underlying key ideas (Properties \ref{itm:h_in_dual_space} and \ref{itm:h_in_dual_space} in Section~\ref{ssec:Overbecks_Attack}).

\subsection{Property \ref{itm:h_in_dual_space}}

\begin{theorem}
	Let $\Code$ be a code defined as in Theorem~\ref{thm:large_q-sum-dim_family} and let $\Code_i := \LambdaOp_i(\Code)$ be its $i^\mathrm{th}$ $q$-sum.
	If $\Code_i \neq \Fqm^n$, then $\Code_i$ does not have a generator matrix of the form
	\begin{align*}
	\begin{bmatrix}
	\pcVec \\
	\pcVec\qpow{1} \\
	\vdots \\
	\pcVec\qpow{j-1}
	\end{bmatrix},
	\end{align*}
	where $\pcVec \in \Fqm^n$ with linearly independent entries and $j \leq n$.
\end{theorem}

\begin{proof}
The $q$-sums of the codes $\Code$ in Theorem~\ref{thm:large_q-sum-dim_family} increase by $\ell+1$ when increasing $i$. If any of the $q$-sums $\Code_i := \LambdaOp_i(\Code)$ had a parity check matrix of the above form, then, by the same arguments as in \cite{Gabidulin_TheoryOfCodes_1985}, $\Code_i$ would have a generator matrix of the form
$[\betaVec^\top,{\betaVec\qpow{1}}^\top, \dots, {\betaVec\qpow{n-j-1}}^\top]^\top$
for some $\betaVec \in \Fqm^n$ with linearly independent entries. This would, however, imply that the first $q$-sum $\LambdaOp_1(\Code_i)$ of $\Code_i$ would have dimension
\begin{align*}
\dim \LambdaOp_1(\Code_i) = \dim \Code_{i+1} = \dim \Code_i +1,
\end{align*}
contradicting the fact that the $q$-sum increases by $\ell+1$.
\end{proof}

\subsection{Property \ref{itm:Lambda_i_n-1}}

\begin{theorem}
	Let $\Code$ be a code defined as in Theorem~\ref{thm:large_q-sum-dim_family} and let $\Code_i := \LambdaOp_i(\Code)$ be its $i^\mathrm{th}$ $q$-sum. Then, we have $\Code_i = \Fqm^n$ or
	\begin{align*}
	\dim \Code_i \leq n-(\ell+1).
	\end{align*}
\end{theorem}

\begin{proof}
This statement follows directly from the dimension of $\Code_i$, which is given in Theorem~\ref{thm:large_q-sum-dim_family}.
\end{proof}

\subsection{An Exponential-Time Attack}\label{ssec:exp-time-attack}

By construction, any $q$-sum (that is not equal to $\Fqm^n$) of a code $\Code$ as in Theorem~\ref{thm:large_q-sum-dim_family} is a subset of a large Gabidulin code $\mathcal{D}$, which is defined as follows:
\begin{align*}
\C_i := \LambdaOp_i(\Code) \subseteq \langle \alphaVec, \alphaVec\qpow{1},\dots,\alphaVec\qpow{n-2} \rangle =: \mathcal{D}.
\end{align*}
Since any non-zero vector $\pcVec$ in the dual code of $\mathcal{D}$ gives a parity check matrix from which the evaluation points $\alphaVec$ can be recovered, such an $\pcVec$ can be found by first determining the $q$-sum $\Code_i$ of $\Code$ of dimension $n-(\ell+1)$ and then searching its dual code $\Code_i^\perp$. Due to $|\Code_i^\perp| = (q^m)^{\ell+1}$, such an attack has work factor
\begin{align}
W_\mathrm{Exp-Att} = \frac{(q^m)^{\ell+1}}{q^m-1} \approx q^{m \ell}. \label{eq:exponential_attack_WF}
\end{align}
Hence, we must choose $m$ and $\ell$ large enough in order to prevent this attack to be efficient.

\section{Example Parameters}

\begin{table*}[!t]
	\caption{Comparison of McEliece, Loidreau, Twisted GPT and QC-MDPC}
	\begin{center}
		\begin{tabular}{l|c|c|c|c|c|c|c|c|c|c|c||c|c|r }
			Method &$q$  & $k$ & $n$ & $m$ & $l$ & $\lambda$ & $s$ & $t$  & $\tau$ & $t_{\text{Loi}}$ & $\lambda^\prime$ & Security level & Rate & Key size  \\
                  \hline  \hline
                  McEliece &2  & 1436 & 1876 & 11 & & & & &  41 & & & 80.04 & 0.77 & 78.98 KB \\
                  \hline
                  Loidreau &2  & 32 & 50 & 50 &  & & & &  & 3 &3  & 80.93 & 0.64 & 3.60 KB \\
                  \hline
                  Twisted GPT &2  & 18 & 26  & 104  & 2  & 6  & 1  & 4  & & & & 83.10 & 0.56 & 3.28 KB \\
                  \hline
                  QC-MDPC & 2   &  4801 & 9602 & & & & &  & & &  & 80.00 & 0.50 & 0.60 KB \\

                  \hline \hline
                  McEliece &2  & 2482 & 3262 & 12 & & & & & 66 & & & 128.02 & 0.76  & 242.00 KB  \\
                  \hline
                  Loidreau &2  & 40 & 64 & 96 & & & &  & & 4 &3 & 139.75 & 0.63 & 11.52 KB \\
                  \hline
                  Twisted GPT &2  & 21 & 33 & 132 & 2 & 8 & 1  & 6 & & & & 138.89 & 0.51 & 6.93 KB \\
                  \hline
                  QC-MDPC & 2   &  9857 & 19714 & & & &  & &  & &  & 128.00 & 0.50 & 1.23 KB \\

                  \hline \hline
                  McEliece&2  & 5318 & 7008 & 13 & & & & & 133 & && 257.47 & 0.76  & 1123.43 KB  \\
                  \hline
                  Loidreau&2  &80 &120  &128  &  & & & & & 4  &5 &261.00  & 0.67  & 51.20 KB  \\
                  \hline
                  Twisted GPT &2   & 32 & 48 & 192 & 2 & 12 & 2 & 8 & & & & 262.75 & 0.53 & 21.50 KB \\
                  \hline
                  QC-MDPC & 2   & 32771  &  65542  & & & &  & & & &  & 256.00 & 0.50 & 4.10 KB \\
		\end{tabular}
	\end{center}
	\label{tab:security_para}
	\hrulefill
\end{table*}

In this Section, the security level, the rate and the keysize of the GPT cryptosystem based on twisted Gabidulin codes is compared with McEliece's cryptosystem based on Goppa codes using list decoding \cite{Barbier_2011}, Loidreau's new rank-metric code-based encryption scheme~\cite{loidreau2016evolution,loidreau2017new} and the QC-MDPC cryptosystem \cite{Misoczki_2013}.

The considered attacks on the new GPT variant based on twisted Gabidulin codes are the syndrome decoding attacks in \cite{ourivski2002new,aragon2017improved}, Gibson's attack \cite{gibson1995severely,gibson1996security}, and the exponential-time attack described in Section~\ref{ssec:exp-time-attack}, cf.~\eqref{eq:exponential_attack_WF}.

Table \ref{tab:security_para} gives parameters for expected work factors of around $2^{80}$, $2^{130}$ and $2^{260}$. The security level of the GPT system which is based on twisted Gabidulin codes is determined by the smallest work factor which is
given by the decoding attack in \cite{ourivski2002new}
for all three cases. We observe that for all work factors McEliece has the highest rate followed by Loidreau, Twisted GPT and QC-MDPC. The results show further that although the keysizes of Twisted GPT and Loidreau are larger than the keysizes of QC-MDPC, they require much smaller key sizes compared to McEliece. 
Since the QC-MDPC scheme gives no guarantee that the cipher can be decrypted, the GPT cryptosystem based on twisted Gabidulin codes should be considered as an alternative of McEliece, Loidreau and QC-MDPC.

\section{Conclusion}

We have shown that a subfamily of twisted Gabidulin resists the Overbeck attack and could therefore be considered for the use in the GPT cryptosystem.
The resulting example key sizes improve upon the original McEliece and the Loidreau's rank-metric cryptosystem.

A drawback of the codes remains that for a small number of twists $\ell$, their $q$-sum dimension is rather low compared to random codes of the same dimension (though larger than the one of a Gabidulin code).
This gives a distinguisher and potentially results in a weakness of the system.
However, we are not aware of an explicit attack that can utilize this distinguisher.
Further research must be conducted in order to investigate this issue.

\bibliographystyle{IEEEtran}
\bibliography{main}

\begin{thebibliography}{10}
\providecommand{\url}[1]{#1}
\csname url@samestyle\endcsname
\providecommand{\newblock}{\relax}
\providecommand{\bibinfo}[2]{#2}
\providecommand{\BIBentrySTDinterwordspacing}{\spaceskip=0pt\relax}
\providecommand{\BIBentryALTinterwordstretchfactor}{4}
\providecommand{\BIBentryALTinterwordspacing}{\spaceskip=\fontdimen2\font plus
\BIBentryALTinterwordstretchfactor\fontdimen3\font minus
  \fontdimen4\font\relax}
\providecommand{\BIBforeignlanguage}[2]{{%
\expandafter\ifx\csname l@#1\endcsname\relax
\typeout{** WARNING: IEEEtran.bst: No hyphenation pattern has been}%
\typeout{** loaded for the language `#1'. Using the pattern for}%
\typeout{** the default language instead.}%
\else
\language=\csname l@#1\endcsname
\fi
#2}}
\providecommand{\BIBdecl}{\relax}
\BIBdecl

\bibitem{Delsarte_1978}
P.~Delsarte, ``Bilinear forms over a finite field with applications to coding
  theory,'' \emph{J. Combinatorial Theory Ser. A}, vol.~25, no.~3, pp.
  226--241, 1978.

\bibitem{Gabidulin_TheoryOfCodes_1985}
E.~M. Gabidulin, ``Theory of codes with maximum rank distance,'' \emph{Probl.
  Inf. Transm.}, vol.~21, no.~1, pp. 3--16, 1985.

\bibitem{Roth_RankCodes_1991}
R.~M. Roth, ``Maximum-rank array codes and their application to crisscross
  error correction,'' \emph{IEEE Trans. Inform. Theory}, vol.~37, no.~2, pp.
  328--336, Mar. 1991.

\bibitem{sheekey2016new}
J.~Sheekey, ``{A New Family of Linear Maximum Rank Distance Codes},''
  \emph{Advances in Mathematics of Communications}, pp. 475--488, 2016.

\bibitem{otal2016explicit}
K.~Otal and F.~{\"O}zbudak, ``{Explicit Construction of Some Non-Gabidulin
  Linear Maximum Rank Distance Codes},'' \emph{Advances in Mathematics of
  Communications}, vol.~10, no.~3, 2016.

\bibitem{beelen2017twisted}
P.~Beelen, S.~Puchinger, and J.~{Rosenkilde n\'e Nielsen}, ``{Twisted
  Reed--Solomon Codes},'' in \emph{IEEE ISIT}, 2017.

\bibitem{beelen2018structural}
P.~Beelen, M.~Bossert, S.~Puchinger, and J.~{Rosenkilde n\'e Nielsen},
  ``{Structural Properties of Twisted Reed--Solomon Codes with Applications to
  Code-Based Cryptography},'' in \emph{IEEE ISIT}, 2018.

\bibitem{puchinger2017further}
S.~Puchinger, J.~{Rosenkilde n{\'e} Nielsen}, and J.~Sheekey, ``{Further
  Generalisations of Twisted Gabidulin Codes},'' 2017.

\bibitem{gabidulin1991ideals}
E.~M. Gabidulin, A.~Paramonov, and O.~Tretjakov, ``{Ideals over a
  Non-Commutative Ring and Their Application in Cryptology},'' in
  \emph{Workshop on the Theory and Application of of Cryptographic
  Techniques}.\hskip 1em plus 0.5em minus 0.4em\relax Springer, 1991, pp.
  482--489.

\bibitem{chabaud1996cryptographic}
F.~Chabaud and J.~Stern, ``The cryptographic security of the syndrome decoding
  problem for rank distance codes,'' in \emph{International Conference on the
  Theory and Application of Cryptology and Information Security}.\hskip 1em
  plus 0.5em minus 0.4em\relax Springer, 1996, pp. 368--381.

\bibitem{ourivski2002new}
A.~V. Ourivski and T.~Johansson, ``New technique for decoding codes in the rank
  metric and its cryptography applications,'' \emph{Problems of Information
  Transmission}, vol.~38, no.~3, pp. 237--246, 2002.

\bibitem{gaborit2016complexity}
P.~Gaborit, O.~Ruatta, and J.~Schrek, ``On the complexity of the rank syndrome
  decoding problem,'' \emph{IEEE Transactions on Information Theory}, vol.~62,
  no.~2, pp. 1006--1019, 2016.

\bibitem{aragon2017improved}
\BIBentryALTinterwordspacing
N.~Aragon, P.~Gaborit, A.~Hauteville, and J.-P. Tillich, ``{Improvement of
  Generic Attacks on the Rank Syndrome Decoding Problem},'' Oct. 2017, working
  paper or preprint. [Online]. Available:
  \url{https://hal.archives-ouvertes.fr/hal-01618464}
\BIBentrySTDinterwordspacing

\bibitem{gabidulin2001modified}
E.~M. Gabidulin and A.~V. Ourivski, ``Modified gpt pkc with right scrambler,''
  \emph{Electronic Notes in Discrete Mathematics}, vol.~6, pp. 168--177, 2001.

\bibitem{gabidulin2003reducible}
E.~M. Gabidulin, A.~V. Ourivski, B.~Honary, and B.~Ammar, ``Reducible rank
  codes and their applications to cryptography,'' \emph{IEEE Transactions on
  Information Theory}, vol.~49, no.~12, pp. 3289--3293, 2003.

\bibitem{loidreau2010designing}
P.~Loidreau, ``Designing a rank metric based mceliece cryptosystem,'' in
  \emph{International Workshop on Post-Quantum Cryptography}.\hskip 1em plus
  0.5em minus 0.4em\relax Springer, 2010, pp. 142--152.

\bibitem{rashwan2011security}
H.~Rashwan, E.~M. Gabidulin, and B.~Honary, ``Security of the gpt cryptosystem
  and its applications to cryptography,'' \emph{Security and Communication
  Networks}, vol.~4, no.~8, pp. 937--946, 2011.

\bibitem{gabidulin2008attacks}
E.~M. Gabidulin, ``Attacks and counter-attacks on the gpt public key
  cryptosystem,'' \emph{Designs, Codes and Cryptography}, vol.~48, no.~2, pp.
  171--177, 2008.

\bibitem{gabidulin2009improving}
E.~M. Gabidulin, H.~Rashwan, and B.~Honary, ``On improving security of gpt
  cryptosystems,'' in \emph{IEEE ISIT}, 2009, pp. 1110--1114.

\bibitem{rashwan2010smart}
H.~Rashwan, E.~M. Gabidulin, and B.~Honary, ``A smart approach for gpt
  cryptosystem based on rank codes,'' in \emph{IEEE ISIT}.\hskip 1em plus 0.5em
  minus 0.4em\relax IEEE, 2010, pp. 2463--2467.

\bibitem{gibson1995severely}
J.~Gibson, ``Severely denting the gabidulin version of the mceliece public key
  cryptosystem,'' \emph{Designs, Codes and Cryptography}, vol.~6, no.~1, pp.
  37--45, 1995.

\bibitem{gibson1996security}
K.~Gibson, ``The security of the gabidulin public key cryptosystem,'' in
  \emph{International Conference on the Theory and Applications of
  Cryptographic Techniques}.\hskip 1em plus 0.5em minus 0.4em\relax Springer,
  1996, pp. 212--223.

\bibitem{overbeck2006extending}
R.~Overbeck, ``Extending gibson’s attacks on the gpt cryptosystem,'' in
  \emph{Coding and Cryptography}.\hskip 1em plus 0.5em minus 0.4em\relax
  Springer, 2006, pp. 178--188.

\bibitem{overbeck2005new}
------, ``A new structural attack for gpt and variants,'' in
  \emph{International Conference on Cryptology in Malaysia}.\hskip 1em plus
  0.5em minus 0.4em\relax Springer, 2005, pp. 50--63.

\bibitem{overbeck2008structural}
------, ``Structural attacks for public key cryptosystems based on gabidulin
  codes,'' \emph{Journal of Cryptology}, vol.~21, no.~2, pp. 280--301, 2008.

\bibitem{horlemann2016considerations}
A.-L. {Horlemann-Trautmann}, K.~Marshall, and J.~Rosenthal, ``Considerations
  for rank-based cryptosystems,'' in \emph{IEEE ISIT}.\hskip 1em plus 0.5em
  minus 0.4em\relax Ieee, 2016, pp. 2544--2548.

\bibitem{otmani2016improved}
A.~Otmani, H.~T. Kalachi, and S.~Ndjeya, ``Improved cryptanalysis of rank
  metric schemes based on gabidulin codes,'' \emph{Designs, Codes and
  Cryptography}, pp. 1--14, 2016.

\bibitem{horlemann2018extension}
A.-L. {Horlemann-Trautmann}, K.~Marshall, and J.~Rosenthal, ``Extension of
  overbeck’s attack for gabidulin-based cryptosystems,'' \emph{Designs, Codes
  and Cryptography}, vol.~86, no.~2, pp. 319--340, 2018.

\bibitem{loidreau2016evolution}
P.~Loidreau, ``An evolution of gpt cryptosystem.''\hskip 1em plus 0.5em minus
  0.4em\relax ACCT, 2016.

\bibitem{loidreau2017new}
------, ``A new rank metric codes based encryption scheme,'' in
  \emph{International Workshop on Post-Quantum Cryptography}.\hskip 1em plus
  0.5em minus 0.4em\relax Springer, 2017, pp. 3--17.

\bibitem{berger2017gabidulin}
T.~P. Berger, P.~Gaborit, and O.~Ruatta, ``Gabidulin matrix codes and their
  application to small ciphertext size cryptosystems,'' in \emph{International
  Conference in Cryptology in India}.\hskip 1em plus 0.5em minus 0.4em\relax
  Springer, 2017, pp. 247--266.

\bibitem{ore1933special}
{\O}.~Ore, ``{On a Special Class of Polynomials},'' \emph{Transactions of the
  American Mathematical Society}, vol.~35, no.~3, pp. 559--584, 1933.

\bibitem{Overbeck-ExtendingGibsonAttack}
R.~Overbeck, ``Extending {G}ibson’s attacks on the {GPT} cryptosystem,''
  \emph{LNCS: Revised Selected Papers of WCC 2005}, vol. 3969, p. 178–188,
  2006.

\bibitem{rosenthal2017decoding}
J.~Rosenthal and T.~Randrianarisoa, ``A decoding algorithm for twisted
  gabidulin codes,'' in \emph{IEEE ISIT}, 2017, pp. 2771--2774.

\bibitem{randrianarisoa2017decoding}
T.~H. Randrianarisoa, ``A decoding algorithm for rank metric codes,''
  \emph{arXiv preprint arXiv:1712.07060}, 2017.

\bibitem{Barbier_2011}
M.~Barbier and P.~S. L.~M. Barreto, ``Key reduction of {McEliece}'s
  cryptosystem using list decoding,'' in \emph{IEEE ISIT}, July 2011, pp.
  2681--2685.

\bibitem{Misoczki_2013}
R.~Misoczki, J.~P. Tillich, N.~Sendrier, and P.~S. L.~M. Barreto,
  ``{MDPC-McEliece: New McEliece Variants from Moderate Density Parity-Check
  codes},'' in \emph{IEEE ISIT}, July 2013, pp. 2069--2073.

\end{thebibliography}

\end{document}